\documentclass[lettersize,journal]{IEEEtran}
\usepackage{amsmath,amsfonts,amssymb}
\usepackage{amsthm}
\usepackage{graphicx}
\usepackage{algorithmicx}
\usepackage{algorithm}
\usepackage{algpseudocode}
\usepackage{textcomp}
\usepackage{makecell}
\usepackage{multirow,multicol}
\usepackage[colorlinks]{hyperref}
\hypersetup{citecolor=blue}
\usepackage[table]{xcolor}
\usepackage{threeparttable}
\usepackage{balance}
\usepackage{booktabs}
\usepackage[mathscr]{eucal}
\newtheorem{theorem}{Theorem}
\newtheorem{corollary}{Corollary}

\newtheorem{proposition}{Proposition}
\theoremstyle{remark}

\usepackage[font=small,skip=2pt]{caption}

\usepackage{enumitem}
\setlist{nosep}
\begin{document}
\IEEEaftertitletext{\vspace{-14mm}}
\title{\LARGE Jointly Correlated Dual-Side Fluid Antenna System}

\author{Zhentian Zhang, Yuanhui Wu, Kai-Kit Wong,~\IEEEmembership{Fellow,~IEEE}, Hao Jiang,~\IEEEmembership{Senior Member,~IEEE}, An Li
\thanks{}
\thanks{Z. Zhang and H. Jiang are with the National Mobile Communications Research Laboratory, Southeast University, Nanjing, 210096, China and H.~Jiang is also with the School of Artificial Intelligence, Nanjing University of Information Science and Technology, Nanjing 210044, China. (e-mail: zhentianzhangzzt@gmail.com, jianghao@nuist.edu.cn).}
\thanks{K. K. Wong is affiliated with the Department of Electronic and Electrical Engineering, University College London, Torrington Place, WC1E 7JE, United Kingdom and he is also affiliated with the Department of Electronic Engineering, Kyung Hee University, Yongin-si, Gyeonggi-do 17104, Korea. (e-mails: kai-kit.wong@ucl.ac.uk)}
\thanks{Y. Wu is with the College of Artificial Intelligence, Nanjing University of Information Science and Technology, 210044, China. (e-mails: 202412621447@nuist.edu.cn)}
\thanks{A. Li is with the School of Information Engineering, Nanchang University, Nanchang 330031, China. (e-mail: lian@ncu.edu.cn)}
\thanks{Corresponding authors: K.-K. Wong (kai-kit.wong@ucl.ac.uk)}}



\maketitle
\begin{abstract}
Fluid antenna systems (FASs) have introduced a new paradigm for wireless system design by revealing how mutual correlation can be exploited to harvest inherent spatial diversity. While existing studies have mainly focused on one-sided FAS configurations, i.e., with FAS deployed at either the transmitter or the receiver, this work investigates the ergodic capacity of a jointly correlated dual-side FAS under statistical eigenmode transmission. Specifically, a jointly correlated dual-side channel model is developed, and the corresponding ergodic capacity together with a tight closed-form upper bound is derived. In addition, the optimal power allocation is studied, and a practical iterative algorithm is proposed for its implementation.
\end{abstract}

\vspace{-4mm}
\section{Introduction}
Fluid antenna systems (FASs) refer to software-controllable fluidic, conductive, or dielectric structures capable of reconfiguring their radiation characteristics to meet communication requirements. By exploiting antenna position reconfigurability, FAS was first introduced into wireless communications in \cite{fas0,FAS_twc_21}. Since then, considerable efforts have been devoted to characterizing the performance limits of FAS-assisted channels \cite{FAS_tutorial}. However, existing asymptotic \cite{FAS_twc_21} and non-asymptotic \cite{performance} studies have mainly focused on single-side FAS, where the fluid antenna is deployed either at the transmitter or the receiver. To extend FAS to broader application scenarios, a dual-side FAS framework is essential, which motivates this work. The main contributions are summarized as follows:
\begin{enumerate}
	\item We establish a jointly correlated dual-side FAS channel model that captures both correlated and diffusive components and applies to a broad range of setups.
	\item Based on this model, we analyze the ergodic capacity and derive a tight upper bound under statistical eigenmode transmission, without requiring instantaneous channel state information.
	\item We further investigate the optimal power allocation for eigenmode transmission and develop an iterative algorithm for practical implementation to approach the capacity.
\end{enumerate}

\emph{Notation:} Boldface lower- and upper-case letters denote vectors and matrices, respectively. 
$(\cdot)^T$, $(\cdot)^H$, $(\cdot)^*$, $\mathrm{tr}(\cdot)$, $\det(\cdot)$, and $\mathbb{E}\{\cdot\}$ denote transpose, Hermitian transpose, conjugate, trace, determinant, and expectation. 
$\boldsymbol{I}_N$ and $\boldsymbol{1}_N$ denote the identity matrix and all-one vector of size $N$. 
$\mathcal{CN}(\boldsymbol{\mu},\boldsymbol{\Sigma})$ denotes the proper complex Gaussian distribution, and $\boldsymbol{A}\succeq\boldsymbol{0}$ means that $\boldsymbol{A}$ is Hermitian positive semidefinite. 
For $\boldsymbol{A}$, $[\boldsymbol{A}]_{m,n}$ denotes its $(m,n)$th entry, $\mathrm{diag}(\boldsymbol{A})$ its diagonal, $\mathrm{diag}(\boldsymbol{a})$ the diagonal matrix with diagonal $\boldsymbol{a}$, $\odot$ the Hadamard product, and $|\boldsymbol{A}|^2$ the element-wise square magnitude. 
$\mathcal{S}_N^{(k)}$ denotes all size-$k$ subsets of $\{1,\ldots,N\}$, and $\boldsymbol{A}_\beta^\alpha$ the submatrix indexed by rows $\alpha$ and columns $\beta$. 
$\boldsymbol{\lambda}_{(i)}$, $\boldsymbol{\Omega}_{(i)}$, and $\boldsymbol{\Omega}_{(m)}^{(i)}$ denote deletion of the $i$th entry, the $i$th column, and the $m$th row together with the $i$th column, respectively. $\mathrm{Per}(\boldsymbol{A})$ denotes the permanent of $\boldsymbol{A}$.

\section{Jointly Correlated Dual-Side FAS Channel}\label{sec:system_model}
Inspired by the jointly correlated framework \cite{jointly_correlated1,jointly_correlated2}, we consider a point-to-point wireless link in which both the transmitter and the receiver are equipped with fluid antennas. The transmit fluid antenna provides $N_t$ preset configurable ports that are uniformly distributed over a linear aperture of normalized length $W_t$, whereas the receive fluid antenna provides $N_r$ preset configurable ports uniformly distributed over a linear aperture of normalized length $W_r$. Throughout the paper, $W_t$ and $W_r$ are normalized by the carrier wavelength. The locations of the transmit and receive ports are respectively given by
\begin{align}
	x_{t,p} &= \frac{(p-1)W_t}{N_t-1}, \quad p=1,2,\ldots,N_t, \\ \nonumber
	x_{r,m} &= \frac{(m-1)W_r}{N_r-1}, \quad m=1,2,\ldots,N_r.
\end{align}

Let $\boldsymbol{H}\in\mathbb{C}^{N_r\times N_t}$ denote the complete port-domain channel matrix, where the $(m,p)$-th entry $[\boldsymbol{H}]_{m,p}$ represents the channel coefficient between the $p$-th transmit port and the $m$-th receive port. To establish a general statistical framework, we first consider a \emph{full-port fluid link}, in which all transmit and receive ports are viewed as available simultaneously. The corresponding input-output relationship is given by
\begin{equation}\label{eq:full_system_model}
	\boldsymbol{y} = \boldsymbol{H}\boldsymbol{x} + \boldsymbol{n},
\end{equation}
where $\boldsymbol{x}\in\mathbb{C}^{N_t\times 1}$ is the transmitted signal vector, $\boldsymbol{y}\in\mathbb{C}^{N_r\times 1}$ is the received signal vector, and $\boldsymbol{n}\sim\mathcal{CN}(\boldsymbol{0},\sigma_\eta^2\boldsymbol{I}_{N_r})$ is the additive white Gaussian noise vector. The transmit covariance matrix is defined as
\begin{equation}\label{eq:tx_covariance}
	\mathbb{E}\{\boldsymbol{x}\boldsymbol{x}^H\} = \frac{P}{N_t}\boldsymbol{Q},
\end{equation}
where $\boldsymbol{Q}\succeq \boldsymbol{0}$ and $\mathrm{tr}(\boldsymbol{Q})=N_t$, such that the total average transmit power equals $P$. We define the average signal-to-noise ratio (SNR) as
\begin{equation}\label{eq:snr_def}
	\rho \triangleq \frac{P}{\sigma_\eta^2},\quad
	\gamma \triangleq \frac{\rho}{N_t}.
\end{equation}
Throughout the work, all types of system follow identical power constraint by \eqref{eq:tx_covariance} and \eqref{eq:snr_def} for fair comparison. 

\vspace{-1mm}
\paragraph{Marginal Port Correlation at Both Ends}

Due to the uniform spacing of the fluid-antenna ports and the continuous aperture structure, the marginal port correlations at the transmitter and receiver can be modeled by two Hermitian Toeplitz matrices, denoted by $\boldsymbol{\Sigma}_t\in\mathbb{C}^{N_t\times N_t}$ and $\boldsymbol{\Sigma}_r\in\mathbb{C}^{N_r\times N_r}$, respectively. Under Clarke-type spatial correlation, the Toeplitz matrices can be accurately specified as
\begin{subequations}
	\begin{align}
	[\boldsymbol{\Sigma}_t]_{p,q}
	&= a_t(p-q),~
	a_t(\ell)=\mathrm{sinc}\!\left(\frac{2\pi \ell W_t}{N_t-1}\right),\label{eq:sigma_t}\\
	[\boldsymbol{\Sigma}_r]_{m,n}
	&= a_r(m-n),~
	a_r(\ell)=\mathrm{sinc}\!\left(\frac{2\pi \ell W_r}{N_r-1}\right). \label{eq:sigma_r}
	\end{align}
\end{subequations}
Their eigenvalue decompositions are given by
\begin{equation}
		\boldsymbol{\Sigma}_t = \boldsymbol{U}_t \boldsymbol{\Lambda}_t \boldsymbol{U}_t^H,\
		\boldsymbol{\Sigma}_r = \boldsymbol{U}_r \boldsymbol{\Lambda}_r \boldsymbol{U}_r^H,
\end{equation}
where $\boldsymbol{U}_t$ and $\boldsymbol{U}_r$ are unitary matrices collecting the transmit and receive eigenmodes, respectively, and $\boldsymbol{\Lambda}_t$ and $\boldsymbol{\Lambda}_r$ are diagonal matrices containing the associated marginal eigenmode powers.

\paragraph{Jointly Correlated Port-Domain Channel Representation}

To capture the \emph{joint} coupling between the transmit fluid modes and the receive fluid modes, we model the complete port-domain channel matrix as
\begin{equation}\label{eq:joint_channel_model}
	\boldsymbol{H}
	=
	\boldsymbol{U}_r \widetilde{\boldsymbol{H}} \boldsymbol{U}_t^H,
\end{equation}
where $\widetilde{\boldsymbol{H}}\in\mathbb{C}^{N_r\times N_t}$ denotes the channel matrix in the joint eigenmode domain. Following the jointly correlated framework, we further write
\begin{equation}\label{eq:eigenmode_channel}
	\widetilde{\boldsymbol{H}}
	=
	\boldsymbol{D} + \boldsymbol{M}\odot \boldsymbol{H}_{0},
\end{equation}
where $\boldsymbol{D}\in\mathbb{C}^{N_r\times N_t}$ denotes the deterministic mean/specular component, $\boldsymbol{M}\in\mathbb{R}_+^{N_r\times N_t}$ is a nonnegative matrix describing the scattering strength between each transmit-receive eigenmode pair, $\boldsymbol{H}_{0}\in\mathbb{C}^{N_r\times N_t}$ has independent zero-mean unit-variance entries, and $\odot$ denotes the Hadamard product.

The key statistical quantity of the proposed dual-sided FAS model is the \emph{joint eigenmode coupling matrix}
\begin{equation}\label{eq:omega_def}
	\boldsymbol{\Omega}
	\triangleq
	\mathbb{E}\!\left\{\widetilde{\boldsymbol{H}}\odot \widetilde{\boldsymbol{H}}^*\right\}
	=
	|\boldsymbol{D}|^2 + \boldsymbol{M}\odot\boldsymbol{M},
\end{equation}
where $|\boldsymbol{D}|^2$ is understood element-wise squared-modulus. The $(i,j)$-th entry $[\boldsymbol{\Omega}]_{i,j}$ quantifies the average channel power coupled from the $j$-th transmit eigenmode to the $i$-th receive eigenmode. Hence, unlike conventional separable models, the proposed representation explicitly characterizes the non-separable interaction between the two fluid antenna apertures through $\boldsymbol{\Omega}$. The marginal receive and transmit eigenmode power profiles are induced by the row and column sums of $\boldsymbol{\Omega}$. To make this relation precise, define the normalized marginal eigenmode power-profile vectors as
\begin{equation}\label{eq:pi_r_pi_t_def}
	\boldsymbol{\pi}_r
	\triangleq
	\frac{\mathrm{diag}(\boldsymbol{\Lambda}_r)}{\mathrm{tr}(\boldsymbol{\Lambda}_r)},~
	\boldsymbol{\pi}_t
	\triangleq
	\frac{\mathrm{diag}(\boldsymbol{\Lambda}_t)}{\mathrm{tr}(\boldsymbol{\Lambda}_t)},
\end{equation}
where $\mathrm{diag}(\boldsymbol{\Lambda}_r)$ and $\mathrm{diag}(\boldsymbol{\Lambda}_t)$ are the vectors collecting the diagonal entries of $\boldsymbol{\Lambda}_r$ and $\boldsymbol{\Lambda}_t$, respectively. Subsequently, the normalized row-sum and column-sum vectors of $\boldsymbol{\Omega}$ satisfy
\begin{equation}\label{eq:omega_col_sum}
	\frac{\boldsymbol{\Omega}\boldsymbol{1}_{N_t}}
	{\boldsymbol{1}_{N_r}^T\boldsymbol{\Omega}\boldsymbol{1}_{N_t}}
	=
	\boldsymbol{\pi}_r,~
	\frac{\boldsymbol{\Omega}^T\boldsymbol{1}_{N_r}}
	{\boldsymbol{1}_{N_r}^T\boldsymbol{\Omega}\boldsymbol{1}_{N_t}}
	=
	\boldsymbol{\pi}_t,
\end{equation}
where $\boldsymbol{1}_{N_t}$ and $\boldsymbol{1}_{N_r}$ are the all-one vectors of lengths $N_t$ and $N_r$, respectively. Hence, the row and column sums of $\boldsymbol{\Omega}$ recover the receive and transmit marginal eigenmode power profiles, respectively, up to a common overall normalization. Therefore, $\boldsymbol{\Sigma}_t$ and $\boldsymbol{\Sigma}_r$ describe only the marginal port correlations, {\em while $\boldsymbol{\Omega}$ captures the full joint correlation structure between the two fluid apertures}. It is worth noting that the proposed model includes the conventional separable correlation model as a special case. In particular, when $\boldsymbol{D}=\boldsymbol{0}$ and $\boldsymbol{M}$ is rank one, or equivalently when $\boldsymbol{\Omega}$ is rank one, the jointly correlated fluid-channel model reduces to a Kronecker-type representation.

\vspace{-4mm}
\section{Capacity and Upper bound}\label{sec:capacity_definition}

\paragraph{Ergodic Capacity}
We assume that the receiver has perfect instantaneous CSI, whereas the transmitter has access only to the {\em channel statistics}, i.e., $\boldsymbol{U}_t$, $\boldsymbol{U}_r$, $\boldsymbol{D}$, $\boldsymbol{M}$, and equivalently $\boldsymbol{\Omega}$. Under this assumption, the ergodic capacity of the full-port dual-sided FAS link is achieved by a zero-mean proper complex Gaussian input and is given by
\begin{equation}\label{eq:full_capacity_def}
	C_{\mathrm{full}}
	=
	\max_{\boldsymbol{Q}\succeq \boldsymbol{0},\, \mathrm{tr}(\boldsymbol{Q})=N_t}
	\mathbb{E}\!\left[
	\log_2
	\det\!\left(
	\boldsymbol{I}_{N_r}
	+
	\gamma \boldsymbol{H}\boldsymbol{Q}\boldsymbol{H}^H
	\right)
	\right].
\end{equation}

Substituting \eqref{eq:joint_channel_model} into \eqref{eq:full_capacity_def}, and following the principle of statistical eigenmode transmission, the transmit covariance matrix can be parameterized as
\begin{equation}\label{eq:Q_stat_eigen}
	\boldsymbol{Q}
	=
	\boldsymbol{U}_t \mathrm{diag}(\boldsymbol{\lambda}) \boldsymbol{U}_t^H,
\end{equation}
where $\boldsymbol{\lambda}=[\lambda_1,\lambda_2,\ldots,\lambda_{N_t}]^T$ satisfies
\begin{equation}\label{eq:lambda_constraint}
	\lambda_i\ge 0,~
	\sum_{i=1}^{N_t}\lambda_i = N_t,~i=1,2,\ldots,N_t.
\end{equation}
Accordingly, the ergodic capacity can be rewritten in the eigenmode domain as
\begin{equation}\label{eq:full_capacity_eigen}
	C_{\mathrm{full}}
	=
	\max_{\boldsymbol{\lambda}\ge \boldsymbol{0},\, \boldsymbol{1}^T\boldsymbol{\lambda}=N_t}
	\mathbb{E}\!\left[
	\log_2
	\det\!\left(
	\boldsymbol{I}_{N_r}
	+
	\gamma \widetilde{\boldsymbol{H}}\mathrm{diag}(\boldsymbol{\lambda})\widetilde{\boldsymbol{H}}^H
	\right)
	\right].
\end{equation}
Equation \eqref{eq:full_capacity_eigen} serves as the {\em exact} ergodic-capacity expression of the jointly correlated dual-fluid antenna channel under statistical CSI at the transmitter.

\paragraph{Single-Port-Pair Selection Capacity}
In a practical fluid-antenna implementation, it is often desirable to activate only one transmit port and one receive port due to hardware and RF-chain limitations. Let $(p^\star,m^\star)$ denote the selected transmit-receive port pair. The resulting scalar input-output relation is
\begin{equation}\label{eq:selection_signal_model}
	y = [\boldsymbol{H}]_{m^\star,p^\star} x + \eta,
\end{equation}
where $x\in\mathbb{C}$ is the transmitted symbol with average power $P$, and $\eta\sim\mathcal{CN}(0,\sigma_\eta^2)$. If the port pair is selected to maximize the instantaneous channel gain, then
\begin{equation}\label{eq:selection_rule}
	(p^\star,m^\star)
	=
	\arg\max_{1\le p\le N_t,\;1\le m\le N_r}
	|[\boldsymbol{H}]_{m,p}|^2.
\end{equation}
The corresponding exact ergodic capacity of the selected dual-fluid link is
\begin{equation}\label{eq:selection_capacity}
	C_{\mathrm{sel}}
	=
	\mathbb{E}\!\left[
	\log_2\!\left(
	1+\rho
	\max_{1\le p\le N_t,\;1\le m\le N_r}
	|[\boldsymbol{H}]_{m,p}|^2
	\right)
	\right].
\end{equation}

Since the practical single-port-pair architecture is a constrained implementation of the full-port fluid link, its ergodic capacity is upper bounded by that of the full-port benchmark, i.e., $C_{\mathrm{sel}} \le C_{\mathrm{full}}$. 

\paragraph{Capacity Upper Bound for the Jointly Correlated Dual-Fluid Channel Channel}
\label{sec:capacity_upper_bound}

Here, we derive a closed-form upper bound on the ergodic capacity of the jointly correlated dual-fluid channel defined in
\eqref{eq:joint_channel_model}--\eqref{eq:omega_col_sum}. The derivation follows the statistical eigenmode transmission principle and converts the expectation of a random determinant into the extended permanent of the joint eigenmode coupling matrix. For notational convenience, let
\begin{equation}
	\boldsymbol{\Lambda} \triangleq \mathrm{diag}(\boldsymbol{\lambda}),~
	\boldsymbol{\lambda}\succeq \boldsymbol{0},~
	\boldsymbol{1}^T\boldsymbol{\lambda}=N_t.
\end{equation}
We also recall the definition of the extended permanent \cite{jointly_correlated1} of an $M\times N$ matrix $\boldsymbol{A}$
\begin{equation}\label{eq:extended_permanent_def_append}
	\underline{\mathrm{Per}}(\boldsymbol{A})
	\triangleq
	\mathrm{Per}\!\left([\boldsymbol{I}_M\ \boldsymbol{A}]\right)
	=
	\mathrm{Per}\!\left([\boldsymbol{I}_N\ \boldsymbol{A}^T]\right).
\end{equation}
Starting from \eqref{eq:full_capacity_eigen}, define
\begin{equation}\label{eq:Cu_lambda_def}
	\widetilde C_{\mathrm{full}}(\boldsymbol{\lambda})
	\triangleq
	\mathbb{E}\!\left[
	\log_2
	\det\!\left(
	\boldsymbol{I}_{N_r}
	+
	\gamma \widetilde{\boldsymbol H}\boldsymbol{\Lambda}\widetilde{\boldsymbol H}^H
	\right)
	\right].
\end{equation}
Then the exact ergodic capacity of the full-port link can be written as
\begin{equation}\label{eq:Cfull_lambda_form}
	C_{\mathrm{full}}
	=
	\max_{\boldsymbol{\lambda}\succeq \boldsymbol{0},\, \boldsymbol{1}^T\boldsymbol{\lambda}=N_t}
	\widetilde C_{\mathrm{full}}(\boldsymbol{\lambda}).
\end{equation}
We now derive a closed-form upper bound for \eqref{eq:Cfull_lambda_form}.
\vspace{-2mm}
\begin{theorem}\label{thm:dual_fluid_upper_bound}
	Under the jointly correlated dual-fluid channel model $		\boldsymbol H=\boldsymbol U_r\widetilde{\boldsymbol H}\boldsymbol U_t^H,~
	\widetilde{\boldsymbol H}=\boldsymbol D+\boldsymbol M\odot \boldsymbol H_{0}$, suppose that the entries of $\boldsymbol H_{0}$ are independent, zero-mean, and unit-variance, and that the mean matrix $\boldsymbol D$ has at most one non-zero entry in each row and each column. Then the ergodic capacity in \eqref{eq:Cfull_lambda_form} is upper bounded as
	\begin{equation}\label{eq:Cfull_upper_bound}
		C_{\mathrm{full}}
		\le
		C_u
		=
		\max_{\boldsymbol{\lambda}\succeq \boldsymbol{0},\, \boldsymbol{1}^T\boldsymbol{\lambda}=N_t}
		\widetilde C_u(\boldsymbol{\lambda}),
	\end{equation}
	where $\widetilde C_u(\boldsymbol{\lambda})
	=
	\log_2
	\underline{\mathrm{Per}}
	\!\left(
	\gamma \boldsymbol{\Omega}\boldsymbol{\Lambda}
	\right)$
	and $\boldsymbol{\Omega}
	=
	\mathbb{E}\!\left\{\widetilde{\boldsymbol H}\odot \widetilde{\boldsymbol H}^*\right\}
	=
	|\boldsymbol D|^2+\boldsymbol M\odot \boldsymbol M.$
\end{theorem}

\begin{proof}
	For any feasible $\boldsymbol{\lambda}$, Jensen's inequality gives
	\begin{equation}\label{eq:jensen_step}
		\widetilde C_{\mathrm{full}}(\boldsymbol{\lambda})
		\le
		\log_2 E(\boldsymbol{\lambda}),
	\end{equation}
	where $E(\boldsymbol{\lambda})
	\triangleq
	\mathbb{E}\!\left\{
	\det\!\left(
	\boldsymbol I_{N_r}
	+
	\gamma \widetilde{\boldsymbol H}\boldsymbol{\Lambda}\widetilde{\boldsymbol H}^H
	\right)
	\right\}.$
Next, using the characteristic-polynomial expansion of the determinant,
	\begin{equation}\label{eq:det_identity}
		\det(\boldsymbol I_M+\boldsymbol X)
		=
		\sum_{k=0}^{\min(M,N)}
		\ \sum_{\hat{\alpha}_k\in \mathcal S_M^{(k)}}
		\det\!\left(\boldsymbol X^{\hat{\alpha}_k}_{\hat{\alpha}_k}\right),
	\end{equation}
	we obtain $E(\boldsymbol{\lambda})
	=
	\mathbb{E}\!\left\{
	\sum\limits_{k=0}^{\min(N_r,N_t)}
	\gamma^k
	\sum\limits_{\hat{\alpha}_k\in \mathcal S_{N_r}^{(k)}}
	\det\!\left(
	(\widetilde{\boldsymbol H}\boldsymbol{\Lambda}\widetilde{\boldsymbol H}^H)^{\hat{\alpha}_k}_{\hat{\alpha}_k}
	\right)
	\right\}$.
	Applying the Cauchy--Binet formula \cite{jointly_correlated1} to each principal minor yields
	\begin{equation}\label{eq:E_expand_2}
		\begin{aligned}
			E(\boldsymbol{\lambda})
			&=
			\sum_{k=0}^{\min(N_r,N_t)}
			\gamma^k
			\sum_{\hat{\alpha}_k\in \mathcal S_{N_r}^{(k)}}
			\sum_{\hat{\beta}_k\in \mathcal S_{N_t}^{(k)}} \\
			&\left(
			\det\!\left(\boldsymbol{\Lambda}^{\hat{\beta}_k}_{\hat{\beta}_k}\right)
			\,
			\mathbb{E}\!\left\{
			\det\!\left(\widetilde{\boldsymbol H}^{\hat{\alpha}_k}_{\hat{\beta}_k}\right)
			\det\!\left(\left(\widetilde{\boldsymbol H}^H\right)^{\hat{\beta}_k}_{\hat{\alpha}_k}\right)
			\right\}\right).
		\end{aligned}
	\end{equation}
	Let $
		\boldsymbol X = \widetilde{\boldsymbol H}^{\hat{\alpha}_k}_{\hat{\beta}_k}$. By construction, the entries of $\boldsymbol X$ are independent, and its second-order moment matrix is
	\begin{equation}
		\mathbb{E}\{\boldsymbol X\odot \boldsymbol X^*\}
		=
		\boldsymbol{\Omega}^{\hat{\alpha}_k}_{\hat{\beta}_k}.
	\end{equation}
	Moreover, the structural assumption on $\boldsymbol D$ guarantees that $\boldsymbol X$ satisfies the condition required by the determinant-to-permanent identity. Therefore,
	\begin{equation}\label{eq:lemma4_apply}
		\mathbb{E}\!\left\{
		\det\!\left(\widetilde{\boldsymbol H}^{\hat{\alpha}_k}_{\hat{\beta}_k}\right)
		\det\!\left(\left(\widetilde{\boldsymbol H}^H\right)^{\hat{\beta}_k}_{\hat{\alpha}_k}\right)
		\right\}
		=
		\mathrm{Per}\!\left(
		\boldsymbol{\Omega}^{\hat{\alpha}_k}_{\hat{\beta}_k}
		\right).
	\end{equation}
	Substituting \eqref{eq:lemma4_apply} into \eqref{eq:E_expand_2}, we obtain
	\begin{equation}\label{eq:E_expand_3}
		E(\boldsymbol{\lambda})
		=
		\sum_{k=0}^{\min(N_r,N_t)}
		\gamma^k
		\sum_{\hat{\alpha}_k\in \mathcal S_{N_r}^{(k)}}
		\sum_{\hat{\beta}_k\in \mathcal S_{N_t}^{(k)}}
		\det\!\left(\boldsymbol{\Lambda}^{\hat{\beta}_k}_{\hat{\beta}_k}\right)
		\,
		\mathrm{Per}\!\left(
		\boldsymbol{\Omega}^{\hat{\alpha}_k}_{\hat{\beta}_k}
		\right).
	\end{equation}
 Since each selected submatrix is square, the scaling property of the permanent gives
	\begin{equation}\label{eq:per_scaling_square}
		\det\!\left(\boldsymbol{\Lambda}^{\hat{\beta}_k}_{\hat{\beta}_k}\right)
		\,
		\mathrm{Per}\!\left(
		\boldsymbol{\Omega}^{\hat{\alpha}_k}_{\hat{\beta}_k}
		\right)
		=
		\mathrm{Per}\!\left(
		(\boldsymbol{\Omega}\boldsymbol{\Lambda})^{\hat{\alpha}_k}_{\hat{\beta}_k}
		\right).
	\end{equation}
	Hence,
	\begin{equation}\label{eq:E_expand_4}
		E(\boldsymbol{\lambda})
		=
		\sum_{k=0}^{\min(N_r,N_t)}
		\gamma^k
		\sum_{\hat{\alpha}_k\in \mathcal S_{N_r}^{(k)}}
		\sum_{\hat{\beta}_k\in \mathcal S_{N_t}^{(k)}}
		\mathrm{Per}\!\left(
		(\boldsymbol{\Omega}\boldsymbol{\Lambda})^{\hat{\alpha}_k}_{\hat{\beta}_k}
		\right).
	\end{equation}
 Using the Laplace-type expansion of the permanent and the definition of the extended permanent, the above summation can be compactly written as
	\begin{equation}\label{eq:E_extended_permanent}
		E(\boldsymbol{\lambda})
		=
		\underline{\mathrm{Per}}
		\!\left(
		\gamma \boldsymbol{\Omega}\boldsymbol{\Lambda}
		\right).
	\end{equation}
	Substituting \eqref{eq:E_extended_permanent} into \eqref{eq:jensen_step} yields $		\widetilde C_{\mathrm{full}}(\boldsymbol{\lambda})
	\le
	\log_2
	\underline{\mathrm{Per}}
	\!\left(
	\gamma \boldsymbol{\Omega}\boldsymbol{\Lambda}
	\right)$.
	Maximizing both sides over all $\boldsymbol{\lambda}$ completes the proof.
\end{proof}
\vspace{-2mm}
Theorem~\ref{thm:dual_fluid_upper_bound} shows that the upper bound depends only on the average SNR $\rho$, the transmit power-allocation vector $\boldsymbol{\lambda}$, and the joint eigenmode coupling matrix $\boldsymbol{\Omega}$. Therefore, {\em once the channel statistics are known}, the upper bound can be evaluated without Monte-Carlo averaging over the instantaneous dual-fluid channel realizations. We next relate a special case of zero-mean Rayleigh dual-fluid channel in Corollary~\ref{cor:selection_upper_bound}.

\begin{corollary}\label{cor:selection_upper_bound}
When the jointly correlated dual-fluid channel is purely scattered, i.e., $\boldsymbol D=\boldsymbol 0$, we have
\begin{equation}
	\boldsymbol{\Omega}=\boldsymbol M\odot \boldsymbol M.
\end{equation}
In this case, the structural condition on the mean matrix is automatically satisfied, and Theorem~\ref{thm:dual_fluid_upper_bound} reduces to
\begin{equation}\label{eq:rayleigh_upper_bound}
	C_{\mathrm{full}}
	\le
	\max_{\boldsymbol{\lambda}\succeq \boldsymbol{0},\, \boldsymbol{1}^T\boldsymbol{\lambda}=N_t}
	\log_2
	\underline{\mathrm{Per}}
	\!\left(
	\gamma (\boldsymbol M\odot \boldsymbol M)\mathrm{diag}(\boldsymbol{\lambda})
	\right).
\end{equation}
This form will be particularly convenient for the analysis of Rayleigh fluid antenna channels and for the development of low-complexity power-allocation strategies based solely on channel statistics.
\end{corollary}
\vspace{-6mm}
\section{Optimal Statistical Power Allocation}
\label{sec:power_allocation}

In this section, we optimize the transmit eigenmode power allocation for the permanent-based capacity upper bound derived in Theorem~\ref{thm:dual_fluid_upper_bound}. Specifically, we consider
\begin{equation}\label{eq:power_alloc_problem_1}
	C_u
	=
	\max_{\boldsymbol{\lambda}\succeq \boldsymbol{0},\,\boldsymbol{1}^T\boldsymbol{\lambda}=N_t}
	\widetilde C_u(\boldsymbol{\lambda}),
\end{equation}
with $\widetilde C_u(\boldsymbol{\lambda})
=
\log_2
\underline{\mathrm{Per}}
\!\left(
\gamma \boldsymbol{\Omega}\mathrm{diag}(\boldsymbol{\lambda})
\right)$,
where $
	\gamma = \frac{\rho}{N_t},~
	\boldsymbol{\lambda}=[\lambda_1,\lambda_2,\ldots,\lambda_{N_t}]^T,~
	\lambda_i\ge 0,\ \sum_{i=1}^{N_t}\lambda_i=N_t$. Since the logarithm is monotonically increasing, the optimization in \eqref{eq:power_alloc_problem_1} is equivalent to
\begin{equation}\label{eq:power_alloc_problem_equiv}
	\max_{\boldsymbol{\lambda}\succeq \boldsymbol{0},\,\boldsymbol{1}^T\boldsymbol{\lambda}=N_t}
	F(\boldsymbol{\lambda}),
\end{equation}
where $F(\boldsymbol{\lambda})
\triangleq
\underline{\mathrm{Per}}
\!\left(
\gamma \boldsymbol{\Omega}\mathrm{diag}(\boldsymbol{\lambda})
\right)$.

\paragraph{Marginal Expansion w.r.t. Each Transmit Eigenmode}

To derive the optimality conditions, we first isolate the dependence of $F(\boldsymbol{\lambda})$ on each $\lambda_i$. Let $\boldsymbol{\lambda}_{(i)}$ denote the vector obtained from $\boldsymbol{\lambda}$ by removing its $i$-th entry. Likewise, let $\boldsymbol{\Omega}_{(i)}$ denote the matrix obtained from $\boldsymbol{\Omega}$ by deleting its $i$-th column, and let $\boldsymbol{\Omega}^{(i)}_{(m)}$ denote the matrix obtained by deleting the $m$-th row and the $i$-th column of $\boldsymbol{\Omega}$. Using the Laplace expansion of the extended permanent with respect to the $i$-th column of $\boldsymbol{\Omega}\mathrm{diag}(\boldsymbol{\lambda})$, we have the following result.

\begin{proposition}\label{prop:affine_lambda_i}
	For any $i\in\{1,2,\ldots,N_t\}$, the function $F(\boldsymbol{\lambda})$ is affine in $\lambda_i$ and can be written as
	\begin{equation}\label{eq:F_affine_i}
		F(\boldsymbol{\lambda})
		=
		F_i^{(0)}(\boldsymbol{\lambda}_{(i)})
		+
		\lambda_i F_i^{(1)}(\boldsymbol{\lambda}_{(i)}),
	\end{equation}
	where $F_i^{(0)}(\boldsymbol{\lambda}_{(i)})
	=
	\underline{\mathrm{Per}}
	\!\left(
	\gamma \boldsymbol{\Omega}_{(i)}\mathrm{diag}(\boldsymbol{\lambda}_{(i)})
	\right)$, and $F_i^{(1)}(\boldsymbol{\lambda}_{(i)})
	=
	\gamma
	\sum_{m=1}^{N_r}
	[\boldsymbol{\Omega}]_{m,i}\,
	\underline{\mathrm{Per}}
	\!\left(
	\gamma \boldsymbol{\Omega}^{(i)}_{(m)}
	\mathrm{diag}(\boldsymbol{\lambda}_{(i)})
	\right)$.
	Consequently,
	\begin{equation} \label{eq:dFdC}
		\frac{\partial F(\boldsymbol{\lambda})}{\partial \lambda_i}
		=
		F_i^{(1)}(\boldsymbol{\lambda}_{(i)}),~
		\frac{\partial \widetilde C_u(\boldsymbol{\lambda})}{\partial \lambda_i}
		=
		\frac{
			F_i^{(1)}(\boldsymbol{\lambda}_{(i)})
		}{
			\ln 2 \, F(\boldsymbol{\lambda})
		}.
	\end{equation}
\end{proposition}

\begin{proof}
	By construction, the $i$-th column of $\boldsymbol{\Omega}\mathrm{diag}(\boldsymbol{\lambda})$ is equal to $\lambda_i [\boldsymbol{\Omega}]_{:,i}$. Since the extended permanent is multilinear with respect to the columns of its matrix argument, $F(\boldsymbol{\lambda})$ is affine in $\lambda_i$. Expanding with respect to that column yields \eqref{eq:F_affine_i}. The first term corresponds to the case where the $i$-th transmit eigenmode is excluded, whereas the second term accounts for all assignments in which the $i$-th transmit eigenmode is used and matched to one of the receive rows. Equations in \eqref{eq:dFdC} then follow directly.
\end{proof}

\paragraph{KKT Conditions}

We next characterize the optimal statistical power allocation through the Karush--Kuhn--Tucker (KKT) conditions.

\begin{proposition}\label{prop:KKT}
	Let $\boldsymbol{\lambda}^\star$ be a locally optimal solution of \eqref{eq:power_alloc_problem_1}. Then there exists a scalar $\mu^\star$ such that
	\begin{equation}\label{eq:KKT_stationarity}
		\frac{
			F_i^{(1)}(\boldsymbol{\lambda}^\star_{(i)})
		}{
			\ln 2 \, F(\boldsymbol{\lambda}^\star)
		}
		-\mu^\star
		\le 0,~i=1,2,\ldots,N_t,
	\end{equation}
	with equality whenever $\lambda_i^\star>0$. Equivalently,
	\begin{equation}\label{eq:KKT_active}
		F_i^{(1)}(\boldsymbol{\lambda}^\star_{(i)})
		=
		\mu^\star \ln 2\, F(\boldsymbol{\lambda}^\star),~\forall\, i \text{ such that } \lambda_i^\star>0,
	\end{equation}
	and
	\begin{equation}\label{eq:KKT_inactive}
		F_i^{(1)}(\boldsymbol{\lambda}^\star_{(i)})
		\le
		\mu^\star \ln 2\, F(\boldsymbol{\lambda}^\star),~ \forall\, i \text{ such that } \lambda_i^\star=0.
	\end{equation}
\end{proposition}

\begin{proof}
	Consider the Lagrangian
	\begin{equation}
		\mathcal L(\boldsymbol{\lambda},\mu,\boldsymbol{\nu})
		=
		\widetilde C_u(\boldsymbol{\lambda})
		-
		\mu\left(\boldsymbol{1}^T\boldsymbol{\lambda}-N_t\right)
		+
		\sum_{i=1}^{N_t}\nu_i \lambda_i,
	\end{equation}
	where $\nu_i\ge 0$ is the dual variable associated with the non-negativity constraint $\lambda_i\ge 0$. Using \eqref{eq:dFdC}, the stationarity condition is
	\begin{equation}
		\frac{
			F_i^{(1)}(\boldsymbol{\lambda}_{(i)})
		}{
			\ln 2 \, F(\boldsymbol{\lambda})
		}
		-\mu+\nu_i=0,~ i=1,\ldots,N_t.
	\end{equation}
	Together with primal feasibility, dual feasibility, and complementary slackness $\nu_i\lambda_i=0$, we obtain
	\eqref{eq:KKT_stationarity}--\eqref{eq:KKT_inactive}.
\end{proof}
Proposition~\ref{prop:KKT} shows that, at an optimal point, all active transmit eigenmodes must have the same marginal contribution to the objective. Therefore, the optimal allocation tends to shift power toward those transmit eigenmodes associated with larger values of
\begin{equation}
	F_i^{(1)}(\boldsymbol{\lambda}_{(i)})
	=
	\gamma
	\sum_{m=1}^{N_r}
	[\boldsymbol{\Omega}]_{m,i}\,
	\underline{\mathrm{Per}}
	\!\left(
	\gamma \boldsymbol{\Omega}^{(i)}_{(m)}
	\mathrm{diag}(\boldsymbol{\lambda}_{(i)})
	\right),
\end{equation}
which can be interpreted as the statistical marginal utility of the $i$-th transmit eigenmode.

\paragraph{Projected-Gradient Statistical Power Allocation}
Although the upper bound in \eqref{eq:power_alloc_problem_1} is available in closed form, {\em the optimization over $\boldsymbol{\lambda}$ generally does not admit a closed-form global solution for an arbitrary jointly correlated dual-fluid channel}. Nevertheless, Proposition~\ref{prop:affine_lambda_i} provides an explicit gradient expression, which enables a projected-gradient algorithm for moderate system dimensions.

Let $\boldsymbol g(\boldsymbol{\lambda})
		=\left[
		\frac{\partial \widetilde C_u(\boldsymbol{\lambda})}{\partial \lambda_1},
		\ldots,
		\frac{\partial \widetilde C_u(\boldsymbol{\lambda})}{\partial \lambda_{N_t}}
		\right]^T$.Given a feasible initialization $\boldsymbol{\lambda}^{(0)}\succeq \boldsymbol{0}$ with $\boldsymbol{1}^T\boldsymbol{\lambda}^{(0)}=N_t$, the iterative update is
		\begin{equation}\label{eq:PG_update_1}
			\boldsymbol z^{(\ell+1)}
			=
			\boldsymbol{\lambda}^{(\ell)}
			+
			\alpha_\ell \boldsymbol g\!\left(\boldsymbol{\lambda}^{(\ell)}\right),
		\end{equation}
		followed by the Euclidean projection onto the simplex
		\begin{equation}\label{eq:PG_update_2}
			\boldsymbol{\lambda}^{(\ell+1)}
			=
			\Pi_{\mathcal S}\!\left(
			\boldsymbol z^{(\ell+1)}
			\right),
		\end{equation}
		where $
			\mathcal S
			=
			\left\{
			\boldsymbol{\lambda}\in\mathbb{R}^{N_t}
			\;\middle|\;
			\lambda_i\ge 0,\ \forall i,\ \boldsymbol{1}^T\boldsymbol{\lambda}=N_t
			\right\},$ and $\alpha_\ell>0$ is a step size. The iterations are terminated when
		\begin{equation}
			\left\|
			\boldsymbol{\lambda}^{(\ell+1)}-\boldsymbol{\lambda}^{(\ell)}
			\right\|_2
			\le \varepsilon,
		\end{equation}
		for some prescribed tolerance $\varepsilon>0$. The complete procedure is summarized in Algorithm~\ref{alg:pg_power_alloc}. The projection in \eqref{eq:PG_update_2} admits the standard closed form
		\begin{equation}\label{eq:simplex_projection}
			\big[\Pi_{\mathcal S}(\boldsymbol z)\big]_i
			=
			\max\{z_i-\tau,0\},
		\end{equation}
		where the scalar $\tau$ is chosen such that
		\begin{equation}
			\sum_{i=1}^{N_t}\max\{z_i-\tau,0\}=N_t.
		\end{equation}
		
		\begin{algorithm}[t]
			\caption{Projected-Gradient Statistical Power Allocation}
			\label{alg:pg_power_alloc}
			\begin{algorithmic}[1]
				\State \textbf{Input:} $\boldsymbol{\Omega}$, $\rho$, $N_t$, tolerance $\varepsilon$
				\State Initialize $\boldsymbol{\lambda}^{(0)}=\boldsymbol{1}_{N_t}$; set $\ell=0$
				
				\Repeat
				\State Compute $
				F\!\left(\boldsymbol{\lambda}^{(\ell)}\right)
				=
				\underline{\mathrm{Per}}
				\!\left(
				\gamma\boldsymbol{\Omega}\mathrm{diag}\!\left(\boldsymbol{\lambda}^{(\ell)}\right)
				\right)$
				\For{$i=1$ to $N_t$}
				\State 
				Compute 
				\[
				F_i^{(1)}\!\left(\boldsymbol{\lambda}^{(\ell)}_{(i)}\right)
				=
				\gamma\sum_{m=1}^{N_r}
				[\boldsymbol{\Omega}]_{m,i}\,
				\underline{\mathrm{Per}}
				\!\left(
				\gamma \boldsymbol{\Omega}^{(i)}_{(m)}
				\mathrm{diag}\!\left(\boldsymbol{\lambda}^{(\ell)}_{(i)}\right)
				\right) \]
				\State Set $g_i^{(\ell)}
				=
				\frac{
					F_i^{(1)}\!\left(\boldsymbol{\lambda}^{(\ell)}_{(i)}\right)
				}{
					\ln 2 \, F\!\left(\boldsymbol{\lambda}^{(\ell)}\right)
				}$
				\EndFor
				\State Update
				$
				\boldsymbol z^{(\ell+1)}
				=
				\boldsymbol{\lambda}^{(\ell)}
				+
				\alpha_\ell \boldsymbol g^{(\ell)}
				$
				\State Project:
				$
				\boldsymbol{\lambda}^{(\ell+1)}
				=
				\Pi_{\mathcal S}\!\left(\boldsymbol z^{(\ell+1)}\right)
				$
				\State $\ell\gets \ell+1$
				\Until{$\|\boldsymbol{\lambda}^{(\ell)}-\boldsymbol{\lambda}^{(\ell-1)}\|_2\le \varepsilon$}
				\State \textbf{Output:} $\boldsymbol{\lambda}^{(\ell)}$
			\end{algorithmic}
		\end{algorithm}

		\paragraph{Structural Insights}
		The KKT characterization also yields useful insights in several important operating regimes.
 {\em --Low-SNR Regime:} When $\rho\to 0$, or equivalently $\gamma\to 0$, the extended permanent admits the first-order expansion
		\begin{equation}\label{eq:low_snr_expand_1}
			F(\boldsymbol{\lambda})
			=
			1
			+
			\gamma
			\sum_{i=1}^{N_t}
			\lambda_i
			\sum_{m=1}^{N_r}
			[\boldsymbol{\Omega}]_{m,i}
			+
			o(\gamma).
		\end{equation}
		Therefore, $\widetilde C_u(\boldsymbol{\lambda})
		=
		\frac{\gamma}{\ln 2}
		\sum_{i=1}^{N_t}
		\lambda_i
		\sum_{m=1}^{N_r}
		[\boldsymbol{\Omega}]_{m,i}
		+
		o(\gamma)$.
		This implies that, in the low-SNR regime, the optimal allocation is asymptotically given by
		\begin{equation}\label{eq:low_snr_opt}
			\lambda_i^\star
			=
			\begin{cases}
				N_t, & i=i^\star,\\
				0, & i\neq i^\star,
			\end{cases}
		\
			i^\star
			=
			\arg\max_{1\le i\le N_t}
			\sum_{m=1}^{N_r}
			[\boldsymbol{\Omega}]_{m,i}.
		\end{equation}
		Hence, all power should be assigned to the statistically strongest transmit eigenmode.
		
 {\em --High-SNR Regime with $N_r\ge N_t$:} When $N_r\ge N_t$ and $\gamma\to\infty$, the dominant term of the extended permanent is the highest-order term of degree $N_t$, namely
		\begin{equation}\label{eq:high_snr_expand_1}
			F(\boldsymbol{\lambda})
			=
			\gamma^{N_t}
			\left(\prod_{i=1}^{N_t}\lambda_i\right)
			\sum_{\hat{\alpha}_{N_t}\in\mathcal S_{N_r}^{(N_t)}}
			\mathrm{Per}
			\!\left(
			\boldsymbol{\Omega}^{\hat{\alpha}_{N_t}}_{1,\ldots,N_t}
			\right)
			+
			o(\gamma^{N_t}).
		\end{equation}
		Since the coefficient in \eqref{eq:high_snr_expand_1} is independent of $\boldsymbol{\lambda}$, maximizing the asymptotic upper bound reduces to $\max_{\lambda_i\ge 0,\ \sum_i\lambda_i=N_t}
		\prod_{i=1}^{N_t}\lambda_i$, whose solution is $\lambda_1^\star=\lambda_2^\star=\cdots=\lambda_{N_t}^\star=1$. Therefore, equal power allocation is asymptotically optimal in the high-SNR regime when $N_r\ge N_t$.

 {\em --Transmit-Symmetric Channels:} If the columns of $\boldsymbol{\Omega}$ are statistically symmetric, then the marginal utilities of all transmit eigenmodes are identical at $\boldsymbol{\lambda}=\boldsymbol{1}_{N_t}$. Hence, the equal-power allocation satisfies the KKT conditions and is a stationary point of \eqref{eq:power_alloc_problem_1}. In particular, for fully symmetric dual-fluid antenna channels, equal power allocation becomes the natural statistics-aware design.
\begin{figure}[t!]
	\centering
	\includegraphics[width=\columnwidth]{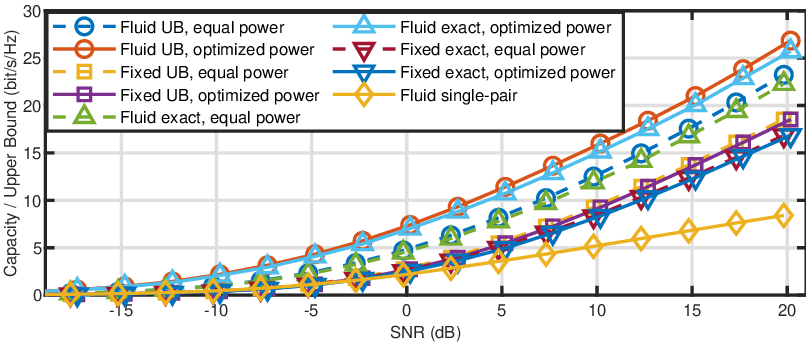}
	\caption{Capacity v.s. SNR. FAS: $(W_t,N_t)=(W_r,N_r)=(1,8)$, Fixed: half-wavelength under same aperture, i.e., $M=2$.}
	\label{fig:capacity_exact_upper}
\end{figure}
\begin{figure}[t!]
	\centering
	\includegraphics[width=\columnwidth]{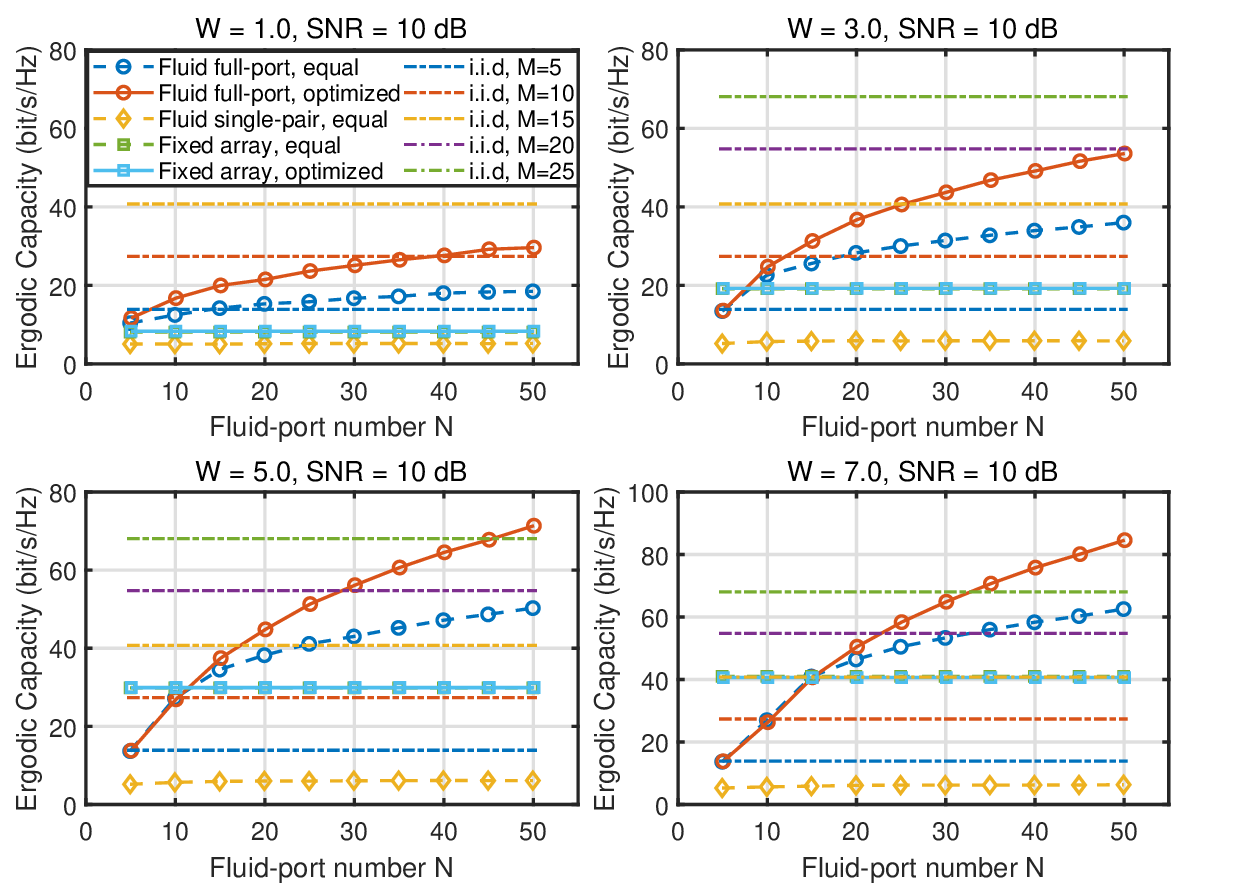}
	\caption{Capacity v.s. full port number. FAS: $N_r=N_t=N$, Fixed: half-wavelength under fixed aperture, i.e., $M=2W+1$; i.i.d antennas, $M\in \{5,10,\ldots,25\}$ with changed aperture.}
	\label{fig:capacity_port_num}
\end{figure}
\begin{figure}[t!]
	\centering
	\includegraphics[width=\columnwidth]{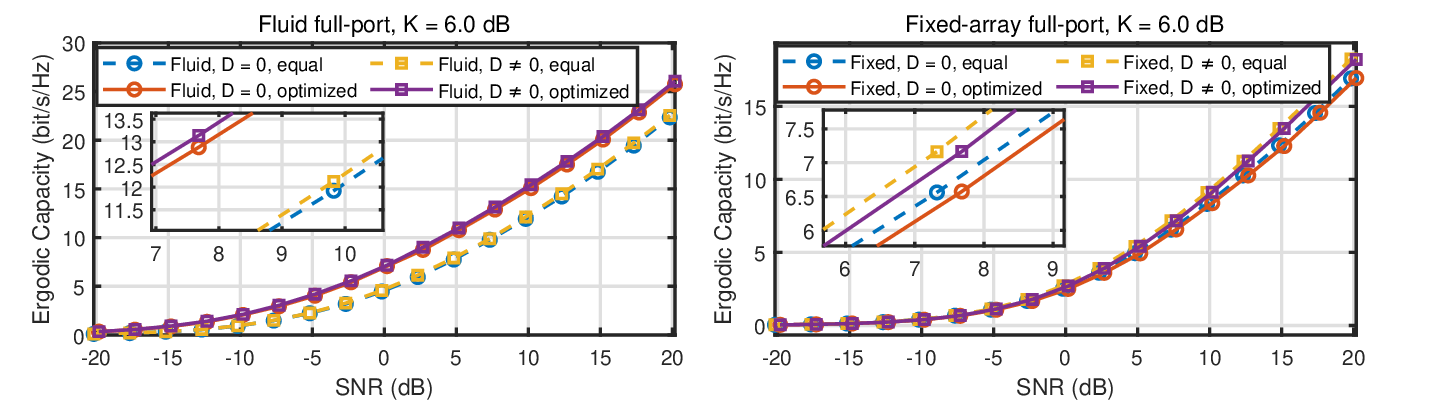}
	\caption{Capacity with/without LOS component. FAS: $(W_t,N_t)=(W_r,N_r)=(1,8)$, Fixed: half-wavelength under same aperture. The power level of non-diffusive against i.i.d component is fixed to $K=6$ dB.}
	\label{fig:capacity_zero_nonzero}
\end{figure}
\vspace{-4mm}
\section{Numerical Results}
This section presents numerical results to verify that dual-side FAS can still achieve substantial capacity even within a compact space. Fig.~\ref{fig:capacity_exact_upper} shows the capacity under different SNRs with only diffusive components, demonstrating that the mutual correlation among adjacent ports does not preclude significant capacity gains. Fig.~\ref{fig:capacity_port_num} further depicts the capacity for different numbers of total ports under purely diffusive components. Owing to the correlation, as $N$ increases, the capacity grows rapidly at first and then gradually levels off. The proposed power allocation scheme is particularly effective in the large-$N$ regime, where the eigenvalues become concentrated in a subset of dominant modes. Finally, Fig.~\ref{fig:capacity_zero_nonzero} examines the impact of non-diffusive components. Interestingly, the capacity of the dual-FAS channel increases slightly in the presence of non-diffusive components, whereas the opposite trend is observed for the fully i.i.d. fixed-antenna system.

\vspace{-4mm}
\section{Conclusion}
In this work, we investigate the exact ergodic capacity and a tight upper bound for dual-side FAS under statistical eigenmode transmission, and develop practical power allocation schemes. Numerical results demonstrate that dual-side FAS can still provide substantial capacity gains in compact spaces.

\end{document}